\documentclass[final,3p,times]{elsarticle}

\usepackage{amssymb}
 \usepackage{amsthm}

\newtheorem{lemma}{Lemma}
\newtheorem{theorem}{Theorem}
\newtheorem{corollary}{Corollary}
\newtheorem{proposition}{Proposition}
\newtheorem{definition}{Definition}
\newtheorem{remark}{Remark}

\journal{Applied Mathematics Letters}

\begin{document}

\begin{frontmatter}



\title{On the Composition of Secret Sharing Schemes Related to Codes}


\author[IMUVa,AGT]{I. M\'arquez-Corbella} 
\author[IMUVa,MA]{E. Mart\'inez-Moro}
\author[IMUVa,AGT]{E. Su\'arez-Canedo}

\address[IMUVa]{Institute of Mathematics University of Valladolid, Castilla, Spain}
\address[AGT]{Algebra, Geometry and Topology Department}
\address[MA]{Applied Mathematics Department}

\begin{abstract}
In this paper we construct a subclass of the composite access structure introduced in \cite{martinez:2004} based on schemes realizing the structure given by the set of codewords of minimal support of linear codes. This class enlarges the iterated threshold class studied in the same paper. Furthermore all the schemes on this paper are ideal (in fact they allow a vector space construction) and we arrived to give a partial answer to a conjecture stated in \cite{martinez:2004}. Finally, as a corollary we proof that all the monotone access structures based on all the minimal supports of a code can be realized by a vector space construction.
\end{abstract}

\begin{keyword}
Cryptography \sep Secret Sharing Schemes \sep Threshold Schemes.
\MSC 94A60

\end{keyword}

\end{frontmatter}


\section{Introduction}
\label{intro}
We will use the following notation. Let $\mathcal P=\{P_i\}_{i=1}^n$ be a set of participants, $\mathcal K$ be the set of all possible keys and $\mathcal S$ be the share sets. \emph{Secret sharing schemes} are used to distribute a secret $K \in \mathcal K$, like a private key of a cryptosystem, among a group of individuals $\mathcal P$, giving to each participant a share from $\mathcal S$, such that only specified subsets of $\mathcal P$ are able to determine the secret $K$ from joining the shares they hold. 
Let $\Gamma \subseteq 2^{\mathcal P}$ be the family of subsets of $\mathcal P$ which are able to reconstructed the secret (i.e. \emph{authorized or qualified subsets}) then $\Gamma$ is called the \emph{access structure} of the scheme. Since $\Gamma$ is presupposed to satisfy the \emph{monotone property} (that is, if $A\subseteq B \subseteq \mathcal P$ and $A\in \Gamma$, then $B \in \Gamma$) then the set of minimal authorized subset of $\Gamma$, denoted by $\Gamma^m$, determines a basis of $\Gamma$. The \emph{dual} of the access structure $\Gamma$  on the set $\mathcal P$ is defined as the access structure form by the subsets whose complements are not authorized, i.e. 
$$\Gamma^{*}= \left\{ A\subseteq \mathcal P \mid \mathcal P \setminus A \notin \Gamma \right\}.$$
A \emph{perfect} sharing scheme avoid unauthorized coalitions to learn any information about the secret.  Ito, Saito and Nishizeki \cite{ito:1989} showed that for any arbitrary monotone collection of authorized set $\Gamma$, there exists a perfect sharing scheme that realizes $\Gamma$. Moreover, a secret sharing scheme is \emph{ideal} if it is perfect and the domain of shares of each user is $\mathcal S$. An access structure $\Gamma$ is called \emph{ideal} if there is an ideal scheme realizing it.
An interesting class of access structure are those admitting a \emph{vector space construction}, this structure is due to Brickell \cite{brickell:1989}. Let $\mathbb F_q$ be a finite field with $q$ elements, an access structure $\Gamma$ on $\mathcal P$ has a \emph{vector space construction} over $\mathbb F_q$ if there exists a map 
$\begin{array}{cccc}\Phi: & \mathcal P & \longrightarrow & \mathbb F_q^d \end{array}$
and a vector $\mathbf v\in \mathbb F_q^d\setminus \{0\}$ such that the vector $\mathbf v$ can be expressed as a linear combination of vectors in the set $\{ \Phi(\mathcal P_i) \mid \mathcal P_i \in A\}$ if and only if $A\in \Gamma$. Schemes realizing this structures are called \emph{vector space secret sharing schemes}. In sake of simplicity and without lost of generality usually $\mathbf v$ is taken to be  the vector $\mathbf e_1=(1,\mathbf 0)$. Unfortunately finding a rule for deciding when an access structure $\Gamma$ admits a vector space construction is still an open problem if the underlying field is not fixed.
The first examples of secret sharing schemes that appeared on the literature were examples of \emph{threshold schemes}. The access structure of an $(t,n)$-threshold scheme is formed by subsets of participants whose cardinality is at least $t$. These schemes were introduced independently by Shamir \cite{shamir:1979} and Blakley \cite{blakley:1979} in $1979$. Shamir's scheme used polynomial interpolation while Blakley's method is based on intersection properties of finite geometries, indeed both ideas where behind or related to the use of Reed-Solomon codes. Threshold schemes are ideal, admit a vector space construction and give the same opportunity to all the participants to access the secret. Indeed taking $n$ different non-zero elements $\alpha_1, \ldots, \alpha_n\in \mathbb F_q$ and $\Phi$ defined by $\Phi(P_i) = ( 1, \alpha_i, \alpha_i^2, \ldots, \alpha_i^{d-1})\in \mathbb F_q^d$ for all $i\in \{1, \ldots, n\}$ then the $(t,n)$-threshold scheme can be seen as a vector space secret sharing schemes. From now on, the expression $(t,n)$ will denote a $(t,n)$-threshold scheme.
In real life, not all participants are in the same hierarchy and they do not have the same privileges to access certain secrets. This idea has been adapted to secret sharing Schemes by various authors. For instance, \emph{multilevel schemes} by Simmons \cite{simmons:1992}, \emph{bipartite structures} by Padr\'o and S\'aez \cite{padro:1998} or \emph{compartmented schemes} by Brickell \cite{brickell:1989}. In this article we will used a special construction of this type of schemes presented in \cite{martinez:2004} called \emph{composition of access structures}. Let $\mathcal P = \mathcal P_1 \cup \ldots \cup \mathcal P_s$ be a partition of $\mathcal P$ into disjoints sets where $\mathcal P_j$ is given by the set $\{  P_1^{(j)}, \ldots,  P_{n_j}^{(j)}\}$ and $n = n_1 + \ldots + n_r$.  Let $\Gamma_0$ be an access structure on $\mathcal P$ and $\Gamma_i$ be an access structure on $\mathcal P_i$ for $i \in \{1, \ldots, r\}$, then the \emph{composite access structure} of $\Gamma_1, \ldots, \Gamma_r$ following $\Gamma_0$, denoted by $\Gamma_0 [\Gamma_1, \ldots, \Gamma_r]$ is defined as follows:
\begin{equation}\label{eq:iterated}
\Gamma_0 [\Gamma_1, \Gamma_2,\ldots , \Gamma_r] = \bigcup_{B\in \Gamma_0} \left\{ A \subseteq \mathcal P \mid A \cap \mathcal P_i \in \Gamma_i \hbox{ for all } \mathcal P_i \in B\right\}. 
\end{equation}
Let us briefly fix the notation and introduce some basic definitions from coding theory. A \emph{linear code} $\mathcal C$ of length $n$ and dimension $k$ over $\mathbb F_q$, or an $[n,k]$ code for short, is a $k$-dimensional subspace of $\mathbb F_q^n$. For every codeword $\mathbf c \in \mathcal C$ its suport is define as its support as a vector in $\mathbb F_q$, i.e. $\mathrm{supp}(\mathbf c) = \left\{ i \mid c_i \neq 0\right\}$. A codeword $\mathbf c$ is a \emph{minimal support codeword} of $\mathcal C$ if it is non-zero and $\mathrm{supp}(\mathbf c)$ is not contained in the support of any other codeword. We will denote by $\mathcal C^m$ the set of codewords of minimal support of $\mathcal C$. Note that describing the set of codewords with minimum hamming weight in an arbitrary linear code is an NP-problem \cite{berlekamp:1978} even if preprocessing is allowed \cite{bruck:1990}.  Some improvements on their computation have been recently made in \cite{Marquez:2011}.
There are several ways to obtain a secret sharing using a linear code $\mathcal C$, we refer the reader to \cite{chunming:2011,massey:1993, renvall:1996}.
It is not difficult to show that a vector space construction is equivalent to a code in the following sense: consider the matrix whose first column is the vector assigned to the dealer and the rest of  columns are the vector assigned to the participants, this matrix can be seen as a parity check matrix of a code $\mathcal C$ and the authorized subsets are those codeword supports containing a non-zero element on the first position.

In this paper we give a slightly different definition to the previous one. We define the access structure related to the $[n,k]$ code $\mathcal C$ over $\mathcal P$ with $|\mathcal P|=n$, and we denote it by  $\Gamma_{\mathcal C}$, as the set
$\Gamma_{\mathcal C} = \left\{ A \subseteq \mathcal P \mid \exists \mathbf c \in \mathcal C\setminus \{0\} ~:~A = \bigcup_{i \in \mathrm{supp}(\mathbf c)} P_i\right\}$. With this definition we study the composite access structures of the form $\Gamma_0 [\Gamma_{\mathcal C_1}, \Gamma_{\mathcal C_2},\ldots , \Gamma_{\mathcal C_r}]$. We enlarge the well known class of iterated threshold structures in \cite{martinez:2004}. The main result is that this structure admits a vector space construction when $\Gamma_0$ admits a vector space construction. This class of structures gives a partial answer to the conjecture in \cite[Open Problem 2]{martinez:2004} and they are more ``natural'' that the one proposed in it since the dealer appears only in one of the components and therefore there is no need of projecting the shares.   As a corollary we obtain that  $\Gamma_{\mathcal C}$ also  admits a vector space construction.

\section{Composition of structures related to linear codes}
\label{sectit}
Let $\{\mathcal C_i\}_{i=1}^r$ be  a set of linear $\mathbb F_q$ codes each one of length $n_i$ and dimension $k_i$ for $i=1,\ldots , r$. For each code $\mathcal C_i$ we define the access structure related to $\mathcal C_i$ over the set of participants $\mathcal P_i=\{ P_1^i,P_2^i,\ldots , P_{n_i}^i\}$ as the set
\begin{equation}
 \Gamma_{\mathcal C_i}=\Gamma_i
\doteqdot
\left\lbrace \{ P_{j_1}^i,\ldots P_{j_s}^i\} \mid \exists \mathbf c\neq 0,\, \mathbf c\in\mathcal C_i\hbox{ such that } \mathrm{supp}(\mathbf c)=\{j_1,\ldots ,j_s\}\right\rbrace .
\end{equation}
That is, the  family of qualified subsets is in one to one correspondence with the supports associated to the codewords of $\mathcal C_i$ and indeed $\Gamma_i^m$ is determined by the minimal support codewords of $\mathcal C_i$.
\begin{definition} Let $\mathcal P_i=\{ P_1^i,P_2^i,\ldots , P_{n_i}^i\}$ be the set of participants related to the code $\mathcal C_i$ for $i=1,\ldots , r$ and consider all of them disjoint. Let $\Gamma_0$ be an access structure over $\{\mathcal P_i\}_{i=1}^r$, we define the access struture $\Gamma_0[\mathcal C_1,\ldots,\mathcal C_s]$ over the set of participants $\mathcal P=\bigsqcup_{i=1}^s \mathcal P_i$ as the composite structure (see Equation~\ref{eq:iterated} for a definition of composite structure)
\begin{equation}
 \Gamma_0[\mathcal C_1,\ldots,\mathcal C_s]=\Gamma_0[\Gamma_{\mathcal C_1},\ldots,\Gamma_{\mathcal C_s}].
\end{equation}
\end{definition}

\begin{remark}
Note that the monotone access structures $\Gamma_{\mathcal C_i}$ are $\mathbb F_q$-matroid representable structures but not in the usual sense (see for example~\cite{brickell:1990}) since they do not have a distinguished participant or a dealer. In our case all the supports in $\mathcal C$ are considered, not only those that include the first coordinate.  Thus,  by definition, it is not obvious that they can be realized by a  vector space construction. We will show in Corollary~\ref{coro} that this last statement is true.
\end{remark}

\begin{remark}\label{remark1}
If each $\mathcal C_i$ is taken to be the Reed-Solomon code $RS(n_i,k_i)$ of parameters $[n_i,k_i]$ and $\Gamma_0$ is a threshold secret sharing scheme then we recover the class of iterated threshold access structures defined in \cite{martinez:2004}. 
\end{remark}
\begin{proposition}
$$\left( \Gamma_0[\mathcal C_1,\ldots,\mathcal C_s]\right)^\star=\Gamma_0^\star[\mathcal C_1^\perp,\ldots,\mathcal C_s^\perp]$$
\end{proposition}
\begin{proof}
We know by \cite[Proposition 2]{martinez:2004} that $\left( \Gamma_0[\Gamma_{\mathcal C_1},\ldots,\Gamma_{\mathcal C_s}]\right)^\star=\Gamma_0^\star[\Gamma_{\mathcal C_1}^\star,\ldots,\Gamma_{\mathcal C_s}^\star]$. But the structure $\Gamma_{\mathcal C_i}^\star$  is representable by a code ($\mathbb F_q$-representable matroid) which is given by its dual code $\mathcal C_i^\perp$ and the result follows. 
\end{proof}
Recall that we will denote by $\Gamma^m$ the minimal qualified subsets in the access structure $\Gamma$ and by $\mathcal C^m$ the subsets of participants in $\Gamma_{\mathcal C}$  related to  minimal codewords of $\mathcal C$.
\begin{proposition}
$$\left( \Gamma_0[\mathcal C_1,\ldots,\mathcal C_s]\right)^m=\Gamma_0^m[\mathcal C_1^m,\ldots,\mathcal C_s^m]$$
\end{proposition}
\begin{proof}
 It follows straightforward from the definitions and \cite[Proposition 1]{martinez:2004}.
\end{proof}
\section{Main Theorem}
\label{MT}
\begin{lemma}\label{lem}
 Let $\mathcal C$ be a $\mathbb F_q$-linear code of parameters $[n,k]$. There exists a $\mathbb F_{q^s}$-linear code $\mathcal C^\prime$ of parameters $[n,k]$ fulfilling the following properties:
\begin{enumerate}
 \item $\Gamma_{\mathcal C}=\Gamma_{\mathcal C^\prime}$.
\item For each minimal support $S\in\{1,\ldots ,n\}$ of $\mathcal C^\prime$ there exists a $\mathbf m\in (\mathcal C^\prime)^m$ with  $\sum_{i=1}^n m_i\neq 0$ and $\mathrm{supp}(\mathbf m)=S$. 
\end{enumerate}
\end{lemma}

\begin{proof}
Let  $\Gamma_{\mathcal C}^m= \{A_1,A_2,\cdots,A_\alpha \}$ be the set of minimal qualified subsets of $\Gamma_{\mathcal C}$ w.r.t. some ordering. Let  $H$ be a  parity check matrix of $\mathcal C$ where $\mathbf h_j$ denotes the $j$-th column with $j=1, \ldots, n$.

By definition $A_1$ is related to at least a codeword support of $\mathcal C$. Assume that all linear combination based on $A_1$ over $\mathbb F_q$ satisfy the following expression:
$$\begin{array}{ccc}
\sum_{j=1}^n \lambda_j \mathbf h_j = 0  & \hbox{ with } & \sum_{j=1}^n \lambda _j = 0
\end{array}.$$
Then we proceed as follows:
\begin{enumerate}
\item Choose an arbitrary linear combination of the above set, say $\lambda_1^{1}, \ldots, \lambda_n^{1}\in \mathbb F_q$, where
$$\begin{array}{cccc}
\lambda_j^{1} \neq 0 \hbox{ if } P_j \in A_1, & \sum_{j=1}^n \lambda_j^{1} \mathbf h_j = 0 & \hbox{ and } & \sum_{j=1}^n \lambda_j^{1} = 0.
\end{array}$$
\item Take a column $\mathbf h_j$ such that $\lambda_j^1 \neq 0$ and define the vector
$$\overline{\mathbf h_j} = \frac{1}{\gamma_1}\mathbf h_j$$
in such a way that $\lambda_j^{1} \gamma_1$ is neither zero nor equal to $-\sum_{i=1}^n \lambda_i^{1} + \lambda_j^{1}$. Note that in the binary case, $q=2$, we need to enlarge the field to some $\mathbb F_{2^{s_1}}$.
\item Define the matrix $H^{1}$ obtained from $H$ by replacing the vector $\mathbf h_j$ by $\overline{\mathbf h_j}$. Observe that $H^{1}$ defines the same linear dependence relations as $H$, since linear dependence behaves well when extending scalars to a field extension, and therefore both matrices realize the same access structure.
\end{enumerate}

At the end of this process we have found a linear combination based on $A_1$ over $\mathbb F_{q^{s_1}}$ such that
$$\begin{array}{ccc}
\sum_{j=1}^{n}\lambda_j^{1} \mathbf h_j^{1} = 0 & \hbox{ and } & \sum_{j=1}^n \lambda_j^{1} \neq 0,
\end{array}$$
where $\mathbf h_j^{1}$ denotes the $j$-th column of the matrix $H^{1}$ for $j=1, \ldots , n$.

Once we have modified the original code and probably the field of definition for the set $A_1$ we check $A_2$. If all linear combination based on $A_2$ over $\mathbb F_{q^{s_1}}$ satisfy the following expression:
$$\begin{array}{ccc}
\sum_{j=1} \lambda_j \mathbf h_j^1 = 0 & \hbox{ with } & \sum_{j=1}^n \lambda_j = 0
\end{array}.$$
Then we proceed as follows (otherwise we skip this step):
\begin{enumerate}
\item Choose an arbitrary linear combination of the above set, say $\lambda_1^2, \ldots, \lambda_n^2$ where
$$\begin{array}{cccc}
\lambda_j^2 \neq 0 \hbox{ if } P_j \in A_2, & \sum_{j=1}^n\lambda_j^2 \mathbf h_j^1 = 0 & \hbox{ and } & \sum_{j=1}^n\lambda_j^2 = 0. \end{array}$$
\item Take a column $\mathbf h_j^1$ such that $\lambda_j^2\neq 0$ and define the vector
$$\overline{\mathbf h_j^1} = \frac{1}{\gamma_2}\mathbf h_j^1$$
in such a way that:
\begin{enumerate}
\item If $P_j \notin A_1$ then $\lambda_j^2 \gamma_2$ is neither zero nor equal to $-\sum_{i=1}^n \lambda_i^2 + \lambda_j$.
\item Otherwise $\lambda_j^2\gamma_2$ has to be different from zero and from the values 
$$\begin{array}{ccc}
-\sum_{i=1}^n\lambda_i^1 + \lambda_j^1 & \hbox{ and } & -\sum_{i=1}^n \lambda_i^2 + \lambda_j^2
\end{array}.$$
\end{enumerate}
\item Define the matrix $H^2$ obtained from $H^1$ by replacing the column $\mathbf h_j^1$ by $\overline{\mathbf h_j^1}$. Again $H^2$ realize the same access structure as $H^1$ and $H$.
\end{enumerate}
Similarly to the previous process, we obtain a linear combination based on $A_2$ over $\mathbb F_{q^{s_2}}$ such that
$$\begin{array}{cccc}
\sum_{j=1}^n \lambda_j^2 \mathbf h_j^2 = 0 & \hbox{ and } & \sum_{j=1}^n \lambda_j^2 \neq 0
\end{array}.$$

Let us now proceed by induction. Suppose that we have a parity check matrix $H^l$ whose code (possibly defined in an extension of the scalars) realizes the structure $\Gamma_{\mathcal C}$ and for each $A_i$ with $i\leq l$ there exists a linear combination of the corresponding rows to the supports of $A_i$ with the sum of the coefficients different from zero. Suppose that for each linear combination based on $A_{l+1}$ over $\mathbb F_{q^{s_l}}$ we have  
$$\begin{array}{ccc}
\sum_{j=1}^{n} \lambda_{j} \mathbf h_j^l=\mathbf 0 & \hbox{ with } &\sum_{j=1}^{n} \lambda_{j}=\mathbf 0.\end{array}$$

Then we choose an arbitrary linear combination of the above set, say $\lambda_1^{l+1}, \ldots, \lambda_n^{l+1}$, we take a column $\mathbf h_j^l$ of $H^l$ corresponding to the support of $A_{l+1}$ such that $\lambda^{l+1}_{j}\neq 0$ and we define
$$\overline{\mathbf h^{l}_{j}}= \frac{1}{\gamma_{l+1}} \mathbf h^l_{j}$$ where $\gamma_{l+1}$ satisfy the following properties: 
\begin{itemize}
\item If $P_j \notin \{A_1, \ldots, A_l\}$ then $\lambda^{l+1}_{j} \cdot \gamma^{l+1}\notin \left\{0,-\sum_{i=1}^{n}\lambda^{l+1}_{i}+\lambda_j^{l+1} \right\}$.
\item If $P_j$ is only in $A_t$ and $A_{l+1}$ with $t=1, \ldots, l$ then 
$$\lambda^{l+1}_{j} \cdot \gamma^{l+1}\notin \left\{0,-\sum_{i=1}^{n}\lambda^{l+1}_{i}+\lambda_j^{l+1} , -\sum_{i=1}^{n}\lambda^{t}_{i}+\lambda_j^{t} 
\right\}.$$
\item $\ldots$
\item If $P_j$ is in $A_{i_1}, \ldots, A_{i_s}$ and $A_{l+1}$ then 
$$\lambda^{l+1}_{j} \cdot \gamma^{l+1}\notin \left\{0,-\sum_{i=1}^{n}\lambda^{l+1}_{i}+\lambda_j^{l+1} , -\sum_{i=1}^{n}\lambda^{i_1}_{i}+\lambda_j^{i_1}, \ldots,  -\sum_{i=1}^{n}\lambda^{i_s}_{i}+\lambda_j^{i_s}
\right\}.$$
\item $\ldots$
\item If $P_j$ is in $A_1, \ldots, A_l, A_{l+1}$ then 
$$\lambda^{l+1}_{j} \cdot \gamma^{l+1}\notin \left\{0,-\sum_{i=1}^{n}\lambda^{l+1}_{i}+\lambda_j^{l+1} , -\sum_{i=1}^{n}\lambda^{1}_{i}+\lambda_j^{1}, \ldots,  -\sum_{i=1}^{n}\lambda^{l}_{i}+\lambda_j^{l}
\right\}.$$
\end{itemize}

The steps above could require to enlarge the field in order to get enough coefficients.
We define $H^{l+1}$ to be the matrix obtained by replacing $\overline{\mathbf h_{j}^{l}}$ by $\mathbf h^{l}_{j}$ in $H^{l}$. $H^{l+1}$ defines the same linear dependence relations as $H^{l},\ldots, H^1$ and $H$. Thus the induction step is proved and we can conclude the proof, i.e. in at most $\alpha$ steps we get a parity check matrix $H^{\alpha}$ defining a code with the required properties.
\end{proof}

\begin{theorem} \label{tm}
If $\Gamma_0$  admits a vector space construction then also $\Gamma_0[\mathcal C_1,\ldots,\mathcal C_s]$ admits a vector space construction.
\end{theorem}
\begin{proof}
Consider the map $\Phi_0:\{\mathcal P_i\}_{i=1}^r\rightarrow \mathbb F_q^d$ that endows $\Gamma_0$ with a vector space construction. For each linear code $\mathcal C_i$ we consider the code $\mathcal C_i^\prime$ that has as parity check matrix the matrix $H_i$ constructed in the proof of Lemma \ref{lem}, probably defined in some field extension of $\mathbb F_q$. We denote by $\mathbf h_j^i$ the $j$-th column of $H_i$. Now we consider the map $\begin{array}{cccc}
 \Phi: & \mathcal P & \longrightarrow & \mathbb F_{q^s}^{d+\sum_{i=1}^s n_i}
\end{array}$
defined by 
$$\Phi(P_j^i) = (\Phi_0(\mathcal P_i),\mathbf 0_{n_1},\ldots ,\mathbf 0_{n_{j-1}}, \underbrace{(\mathbf h_j^i)^t}_{j+1\hbox{-th position}},  \mathbf 0_{n_{j+1}},\ldots ,\mathbf 0_{n_{s}}),$$
where $\mathbf 0_l$ denotes the zero vector of length $l$. We shall prove that $\Phi$ endows $\Gamma=\Gamma_0[\mathcal C_1^\prime,\ldots,\mathcal C_s^\prime]$ with a vector space construction, and therefore also  $\Gamma_0[\mathcal C_1,\ldots,\mathcal C_s]$ has a vector space construction since they define the same access structure by Lemma~\ref{lem}. 
 Let $A\in\Gamma$ be a qualified set and $B=\{\mathcal P_i\mid \mathcal P_i\cap A\in \Gamma_i\}\in\Gamma_0$. Let $A_i=\{P_{j_1}^i,\ldots, P_{j_{l_i}}^i \}\neq \emptyset$ be the set $A\cap \mathcal P_i$ and suppose that it is a minimal qualified set (otherwise it always contains one).  Thus the vectors $\{\mathbf h_{j_1}^i,\ldots, \mathbf h_{j_{l_i}}^i \}$ are linearly dependent and all  subsets of them of cardinality $l_i-1$ are linearly independent. By Lemma \ref{lem} we have that there exist a codeword in $\mathcal C_i^\prime$ given by $(0,\ldots, 0, \lambda_{j_1}^i, 0,\ldots, 0, \lambda_{j_{l_i}}^i, 0,\ldots, 0)$ such that 
$\mathbf 0=\sum_{k=1}^{l_i} \lambda_{j_k}^i\mathbf h_{j_k}^i$ and $\sum_{k=1}^{l_i} \lambda_{j_k}^i\neq 0$.
Thus for each $\mathcal P_i\in B$ the following non-zero vector
$$\mathbf 0\neq  \sum_{k=1}^{l_i} \lambda_{j_k}^i \Phi(P_{j_k}^i)= \left( \sum_{k=1}^{l_i} \lambda_{j_k}^i  \Phi_0(\mathcal P_i),\mathbf 0,\ldots,\mathbf 0\right)$$
belongs to $\left\langle \Phi(A) \right\rangle$, and since $\Phi_0$ defines a vector space structure on $\Gamma_0$ then 
$$\mathbf e_1\in \left\langle \sum_{k=1}^{l_i} \lambda_{j_k}^i \Phi_0(\mathcal P_i) \right\rangle_{\mathcal P_i\in B}$$ 
and we have that 
$(\mathbf e_1, \mathbf 0,\ldots,\mathbf 0)\in \left\langle \Phi(A) \right\rangle$.

On the other hand, let now $A\subseteq \mathcal P$ be a participant set such that 
$\mathbf e_1\in \left\langle \Phi(A) \right\rangle$. Then $\mathbf e_1\in \left\langle \Phi_0(B) \right\rangle$ and for each $\mathcal P_i\in B$ if $A_i=A\cap \mathcal P_i$ then $\mathbf 0\in \left\langle \pi_{i}(\Phi(B) ) \right\rangle$ where $\pi_{i}$ is the restriction of $\Phi(B)$ to the interval $\left[\begin{array}{cc} d+1+\sum_{j=1}^{i-1} n_j, & d+\sum_{j=1}^{i} n_j\end{array}\right]$. Therefore there exists a codeword in $\mathcal C^\prime_i$ with support corresponding to the participants of the set $A_i=A\cap \mathcal P_i$ for each $\mathcal P_i\in B$. 
\end{proof}
\begin{corollary}\label{coro}
$\Gamma_{\mathcal C}$ admits a vector space construction.
\end{corollary}
\begin{proof}
Note that $\Gamma_{\mathcal C}=(1,1)[\mathcal C]$ so we can apply the above theorem.
\end{proof}

\section*{Acknowledgments}
\noindent  The first two authors are partially supported by Spanish MCINN under project MTM2007-64704. First author research is also supported by a FPU grant AP2008-01598 by Spanish MEC. Second author is also supported by Spanish MCINN under project MTM2010-21580-C02-02.





\bibliographystyle{plain}
\bibliography{emilio}







\end{document}